%% file: RQvsHQ.tex
\documentclass[10pt,conference,twocolumn]{IEEEtran}

\usepackage{amssymb, amsmath, amsthm}
\usepackage{color,graphicx,comment}
\usepackage{epstopdf}
\usepackage{caption,cite}
\usepackage{subcaption}
\usepackage{psfrag}
\usepackage{bbm}

\usepackage{tikz}
\usetikzlibrary{shapes,arrows,fit}
\usetikzlibrary{plotmarks}
\usetikzlibrary{positioning}
\usetikzlibrary{decorations.markings}
\usepackage{pgfplots}
 \pgfplotsset{compat=newest}

\title{Isn't Hybrid ARQ Sufficient?}
\author{
\authorblockN{Michael Heindlmaier\IEEEauthorrefmark{1} and Emina Soljanin\IEEEauthorrefmark{2}}
\authorblockA{
\IEEEauthorrefmark{1}Institute for Communications Engineering, Technische Universit\"at M\"unchen,
Munich, Germany\\
\IEEEauthorrefmark{2}Bell Labs, Alcatel-Lucent,
Murray Hill NJ 07974, USA\\
Email: michael.heindlmaier@tum.de, emina@alcatel-lucent.com}
}

\newtheorem{thm}{Theorem}

\begin{document}
\maketitle
\begin{abstract}
In practical systems, reliable communication is often accomplished by coding at different network layers. We question the necessity of this approach and examine when it can be beneficial. Through conceptually simple probabilistic models (based on coin tossing), we argue  that multicast scenarios and  protocol restrictions may make concatenated multi-layer coding  preferable to physical layer coding  alone, which is mostly not the case in point-to-point communications.
\end{abstract}

\input{Intro}
\input{SystemModel}

\input{RApps}
\input{RAms}

\section{Conclusion and Further Work}
We investigated the relationship of Hybrid-ARQ and rateless packet level coding schemes.
Using a coin tossing model we observe the well-accepted fact that Hybrid-ARQ alone is advantageous in point-to-point scenarios. However, for multicast scenarios with many users, we observe that a packet level coding can outperform schemes that rely on Hybrid-ARQ only. In this context, we reviewed approximations for order statistics of negative binomial random variables and proposed improved approximations.

For future work, we target to extend these results to more general cases and add refined model assumptions.

\begin{appendix}
\input{MaxNegBinom}
\end{appendix}
\bibliographystyle{IEEEtran}
\bibliography{IEEEabrv,references}
\end{document}

%% file: Intro.tex
\section{Introduction}
\subsection{Motivation}
In packet-based data networks, large files are usually segmented into smaller blocks that are put into transport packets. Packet losses occur because of the physical channel and other limitations, such as processing power and buffer space. In current wireless systems, reliable communication is accomplished by coding at different network layers.

At the physical layer, a special transmission scheme, known as incremental redundancy Hybrid ARQ (IR-HARQ), which combines the conventional ARQ with error correction, has been in use since the appearance of 3G wireless technology (see, for example \cite{HARQart} for an overview).
IR-HARQ schemes adapt their error code redundancy, based on the receiver's feedback, to varying channel conditions, and thus achieve better throughput performance than ordinary ARQ.

In broadcast/multicast applications from a single sender to many receivers, however, it is costly for the sender to collect and respond to individual receiver feedbacks, and thus HARQ schemes are disabled and packet losses are inevitable. With the rapid increase in multicast streaming applications, we see more and more proposals for packet level rateless erasure coding. A number of these schemes have already been standardized and are currently being implemented and deployed,
e.g., Raptor codes for LTE eMBMS \cite{RCmonograph}. At the packet level, rateless codes enable efficient communications over multiple, unknown erasure channels,  by asymptotically and simultaneously achieving the channel capacity at all erasure rates.

Although IR-HARQ, as a unicast technique, is disabled in multicast systems,  physical layer coding at some chosen fixed rate remains. When this code fails to decode the noisy version of a data packet at the physical layer, the packet is declared erased, and data recovery is left to the packet level code.
The rate of the physical layer code affects the successful-transmission time both positively and negatively. Increasing the rate, increases the number of channel uses, but lowers the packet erasure rate, which in turn means that fewer application layer rateless code packets will be sufficient for successful data transmission.

In practice, packet level coding schemes are not rateless. Standardized coding schemes start with a time-limited broadcast {\it delivery phase} by a content server using systematic Raptor codes with some
fixed redundancy. This phase is often too short for all broadcast clients to collect a sufficient number of Raptor code symbols to be able to decode. Hence, once the delivery (broadcast) session expires, a unicast based file repair mechanism becomes available to the users who had experienced bad channel realizations and were not able to decode. These users enter the {\it repair phase,} during which the missing data is delivered by dedicated repair servers through unicasts with no packet level coding.

\subsection{Related Work}
Because of its practical relevance, this problem has been addressed in multiple ways, but only to a limited extent. Information theoretic analysis for single user scenarios suggest that channel codes be implemented entirely at the physical layer, where they can most efficiently combat fading 
\cite{berger2008optimizing,courtade2011optimal}. Furthermore, concatenated scheme involving coding at the physical as well as at a higher layer would be suboptimal.

This and related setups have also been investigated in \cite{koller2011optimal, 6068196, cui2009achievable, teerapittayanon2012network, vehkapera2005throughput,xiao2011cross}.
Reference \cite{koller2011optimal} studies the code rate tradeoff for both unicast and multicast setups where individual links are modeled as binary symmetric channels. Without packetization constraints, it turns out that in this case pure physical layer coding is optimal. In practice systems however, packetization is inevitable to some extent.

Reference \cite{6068196} investigates the interplay between rate allocation and ARQ for point-to-point links, but do not consider coding at the packet level.
Reference \cite{cui2009achievable} investigates the tradeoff between physical layer and network layer rate allocation for networks in order to improve throughput, but does not study the possibility of HARQ.
In an experimental setup, \cite{teerapittayanon2012network} studies the application of packet layer coding over a WiMAX point-to-point link: The best performance was observed with HARQ disabled and relying only on packet coding as reliability mechanism. 
The works \cite{vehkapera2005throughput, xiao2011cross} compute the optimal physical layer code rate when packet level codes are operated in a rateless fashion, but do not study the impact of HARQ.

\subsection{Our Approach}
Our approach is to simplify the transmission model in order to better understand the interplay between the inner and the outer code for several single and multiuser scenarios that are motivated by current practice. The physical layer is modeled as a memoryless symbol erasure channel with fixed erasure probability $\epsilon_s$ which is given by the channel conditions.
Successive symbol transmissions can thus be interpreted as flipping a biased coin, where the bias is given by system constraints and cannot be manipulated.

A packet contains $k_s$ symbols and is successfully recovered if any $k_s$ symbols are received, otherwise erased. The number of transmitted symbols, $n_s$, determines the coding rate of the physical layer. This rate can be adapted by means of HARQ, so $n_s$ can vary from packet to packet in principle.
For a fixed $n_s$ successive packets are erased independently. Again, successive packet transmissions can thus be interpreted as flipping a biased coin where the bias can be adapted with the choice of $n_s$.

This simple coin tossing model is easy to understand but does not accurately represent the characteristics of the wireless channel.
We still believe that the following essential feature is well approximated:
In wireless systems with opportunistic link adaptation, the number of symbols needed to be transmitted for successful packet decoding is random because of imperfect channel knowledge. This is also captured by the coin tossing model.

The model does not capture the delay that is introduced by a practical implementation of HARQ schemes due to feedback delay or packet scheduling issues.
These effects also have an impact on the delay performance \cite{larmo2009lte}.
The goal of this work is to show situations where, even without delay aspects, optimal coding operations become nontrivial.

This paper is organized as follows: In Sec.~\ref{sec:SM}, we present the system model under study. Different schemes for point-to-point scenarios are analyzed in Sec.~\ref{sec:pps}. We focus on the multiuser case in Sec.~\ref{sec:ms}.

%% file: SystemModel.tex

\section{System Model and Problem Formulation\label{sec:SM}}
\subsection{Data, Channel, and Code Models}
We are transmitting a large data file of $M$ chunks, each chunk has $k_p$ packets, and each packet $k_s$ channel symbols. Consequently, each chunk has $k_p\cdot k_s$ physical channel symbols. In transmission, the $k_s$ data symbols of each packet may be protected by a channel code of length $n_s$ at the physical layer, and the $k_p$ data packets in each chunk may be protected by a packet level code of length $n_p$. This is visualized in Fig.~\ref{fig:packetization}.
 \begin{figure}[ht]
  \centering
  \psfrag{a}{$k_p$}
    \psfrag{b}{$k_s$}
      \psfrag{c}{$n_s$}
       \psfrag{d}{$n_p$}
       \psfrag{ch}{chunks}
       \psfrag{pa}{packets}
       \psfrag{sy}{symbols}
  \includegraphics[width=0.95\columnwidth]{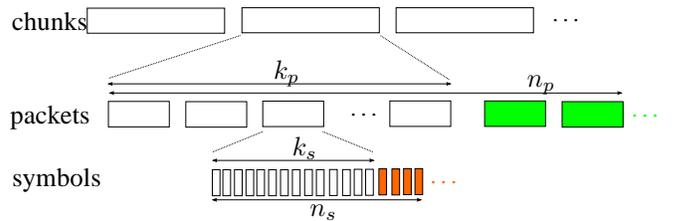}
  \caption{System model for packetization}
  \label{fig:packetization}
 \end{figure}
 
The physical layer is modeled by a memoryless erasure channel with symbol erasure probability $\epsilon_s$, assumed to be given. Coding is assumed to be such that a
packet is successfully received if (any) $k_s$ symbols are received, and a chunk is received if (any) $k_p$ packets are received.

\subsection{User, Redundancy, and Feedback Models}
We will be considering {\it point-to-point} as well as {\it multicast} scenarios to $u$ users.
In a point-to-point scenario, either the inner or the outer code may be rateless \cite{RCmonograph}, that is, coded symbols and packets may be sent until the packets and chunks are decoded. Note that if the inner (physical layer) code is rateless, the outer (packet level) code is not necessary.

In multicast, the inner (physical channel) code may have a fixed rate, or the coded symbols may be transmitted until some $\ell$ users are able to decode. Note that $\ell\le u$, and if $\ell=u$, then the outer (packet level)  code is not needed.

We assume that each receiver can send a feedback signal at any moment on either level.  The feedback is instantaneous, noiseless, and signifies only the completion of decoding. Under this assumption, the transmitter can send coded symbols (i.e., incremental redundancy) one symbol at the time, and does not have to send coded symbols that are not necessary for decoding.

\subsection{Transmission Objectives}
\subsubsection{Expected File/Chunk Download Time}
The goal of each user is to eventually download the file/chunk. Therefore, the most natural transmission objective is to minimize the expected time at which all user have successfully downloaded the file, which in general requires rateless transmission. In this paper, we will be mainly concerned with this objective, and our goal will be to find out
at which level the rateless transmission should take place. Moreover, if the rateless coding is done only at the packet level, we are interested to find out which coding rate at the physical layer minimizes the overall multicast transmission time in this case.

\subsubsection{Probability of Decoding Within Some Allocated Time}
In practice, standardized coding schemes start with a time-limited broadcast delivery phase by a content server using systematic Raptor codes with some fixed redundancy. This phase is often too short for all broadcast clients to collect a sufficient number of Raptor code symbols to be able to decode. Therefore, another reasonable objective is to maximize probability of decoding under a delivery-time constraint, which in turn maximizes the expected number of users who are able to download the file in within this time and will not have to go through the file repair phase.

%% file: RApps.tex
\section{Point-to-Point Scenarios}
\label{sec:pps}
Motivated by some common state-of-the-art schemes, we consider  the following three point-to-point scenarios depending on the type of coding used at the (inner) physical layer:

\setcounter{subsubsection}{0}
\subsubsection{Infinite Incremental Redundancy (IIR)} This is an idealized rateless scheme where transmission at the physical layer of coded symbols is done until $k_s$ symbols are received.  Consequently, no packet gets erased, and thus no coding or retransmission at the packet level is needed.

\subsubsection{Fixed Redundancy (FR)} This scheme is motivated by the one in the LTE eMBMS protocol, where a fixed rate code $(n_s,k_s)$ is used at the physical layer. Consequently, some packets will be erased, and coding or retransmission at the packet level is required.

\subsubsection{Finite Incremental Redundancy (FIR)} This scheme most closely resembles the current IR-HARQ implementations. Here again, a fixed rate code $(n_s,k_s)$ is used at the physical layer,
and transmission of coded symbols is done until either $k_s$ symbols are received or $n_s$ coded symbols are used up (whichever happens first). Consequently, some packets will be erased, and coding or retransmission at the packet level is required.

Recall that, since our objective is to ultimately decode the chunk, either the inner or the outer code must be rateless, that is, must allow unlimited number of transmissions until decoding can be performed. Let $T_i^s$,  be the 
number of channel symbols used for the $i$-th packet transmission, and $T^p$ the number of (coded) packets that have to be sent to decode the chunk. Then decoding of the chunk requires
\begin{align}
 \sum_{i=1}^{T^p} T_i^s \label{eq:all_slots}
\end{align}
channel transmissions. Note that, depending on the scenario, either $T_i^s$ or
$T^p$ or both can be random variables.

We next describe the three point-to-point scenarios, derive the expected number of channel transmissions necessary for decoding, and provide some quantitative examples. We conclude the section by showing that using rateless coding at the inner rather than at the outer layer results in fewer number of transmissions on average.

\subsection{Infinite Incremental Redundancy (IIR) Schme}
In the IIR scheme, the code is rateless at the physical layer,  meaning that coded symbols are transmitted
over the channel with the erasure probability $\epsilon_s$
until $k_s$ of them are received, at which point the receiver can decode the data.
Therefore, the number of transmissions is a negative binomial random variable with parameters $(k_s,1-\epsilon_s)$\footnote{$1-\epsilon_s$ is the probability of success in the associated Bernoulli trials.}, which we will denote by $NB(k_s,1-\epsilon_s)$
The number of channel transmissions required to successfully decode $T^p=k_p$ packets is the sum of these random variables, and thus itself a negative binomial random variable with parameters $(k_p\cdot k_s,1-\epsilon_s)$. We denote this time by $T^{\text{IIR}}$. Its expected value is given by
\begin{equation}
\mathbb{E}\bigl[T^{\text{IIR}}\bigr]=\frac{k_p\cdot k_s}{1-\epsilon_s}
\label{eq:TIIR}
\end{equation}

\subsection{Fixed Redundancy (FR) Scheme}
In the FR scheme, a fixed-rate code $(n_s,k_s)$ is used at the physical layer, and thus it always takes  $T_i^s=n_s$ channel transmissions per data packet, at which point the packet is either successfully decoded (when $k_s$ or more channel transmissions are successful) or erased with probability $\epsilon_p$, given by
\begin{equation}
\epsilon_p = \sum\limits_{j=k_s+1}^{n_s}\binom{n_s}{j}\epsilon_s^j(1-\epsilon_s)^{(n_s-j)}.
\label{eq:ep}
\end{equation}
The FR scheme is rateless at the packet level. Therefore, the number of packet transmissions over the channel with the erasure probability $\epsilon_p$ until $k_p$ transmissions are successful (at which point the receiver can decode the data) is a negative binomial random variable  with parameters $(k_p,1-\epsilon_p)$. Since each packet transmission takes exactly $n_s$ channel transmissions, the number of channel transmissions $T^{\text{FR}}$ required to decoded the file is not a negative binomial random variable, but its expected value is given by
\begin{equation}
\mathbb{E}\bigl[T^{\text{FR}}\bigr]=\frac{k_p\cdot n_s}{1-\epsilon_p}.
\label{eq:TFR}
\end{equation}

Note that it is not immediately clear how (\ref{eq:TFR}) compares to (\ref{eq:TIIR}) since $n_s\ge k_s$ but $\epsilon_s\ge \epsilon_p$. However, we show below that $\mathbb{E}\bigl[T^{\text{FR}}\bigr]\ge \mathbb{E}\bigl[T^{\text{IIR}}\bigr]$, that is, using  rateless coding at the inner rather than at the outer layer results in fewer number of transmissions on average.

\subsection{IIR vs.\ FR}
\begin{thm} 
$\mathbb{E}\bigl[T^{\text{FR}}\bigr]\ge \mathbb{E}\bigl[T^{\text{IIR}}\bigr]$
\end{thm}

\begin{proof}
Consider the non-negative random variable $S$ corresponding to the number of successfully received symbols after $n_s$ transmission. Note that $S$ is binomial with parameters $n_s$ and $1-\epsilon_s$, and therefore its expectation is $n_s\cdot(1-\epsilon_s)$.  By the definition of $\epsilon_p$ and the
Markov's inequality (see e.g., \cite[p.~116]{Flajolet:2009}), we have
\begin{align}
1-\epsilon_p & = P(S\ge k_s)\le \frac{n_s\cdot(1-\epsilon_s)}{k_s},
\end{align}
which gives
\[
\frac{k_s}{1-\epsilon_s} \le \frac{n_s}{1-\epsilon_p},
\]
and the claim follows from \eqref{eq:TIIR} and \eqref{eq:TFR}.
\end{proof}

Therefore rateless coding at the inner rather than at the outer layer results in fewer channel uses on average for chunk download. However, when $\epsilon_s$ is known, then 
$n_s$ can be optimized to minimize $n_s/(1-\epsilon_p)$, then $\mathbb E \bigl[ T^{\text{IIR}}\bigr]$ and $\mathbb E \bigl[ T^{\text{FR}}\bigr]$ do not differ much, as shown by an example in Fig.~\ref{fig:delay_somp_12}.
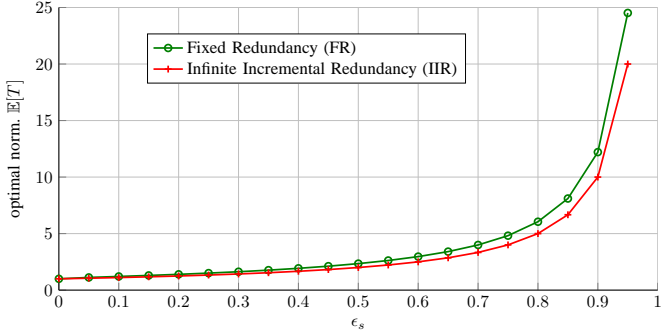
\begin{figure}[hbt]
\centering
\input{./figs/delay_over_eps_c_kc=100.tex}
\caption{Normalized delay for FR and IIR for $k_s = 100$. $n_s$ was chosen to minimize 
$\mathbb E [ T^{\text{FR}}]$.}
\label{fig:delay_somp_12}
\end{figure}

Note instantaneous and perfect feedback should be possible after each symbol for IIR, and therefore, when the channel is known and stays constant, FR beats IIR in practice. However,
in current practical wireless scenarios, the channel is unpredictable, and a HARQ scheme in use is similar to the one we consider in the next section. 

\subsection{Finite Incremental Redundancy (FIR) Scheme}

In the FIR scheme,  both $T^p$ or $T_i^s$ are random variables in contrast to the previous schemes:
At the physical layer at most $n_s$ symbols are transmitted for each transmission, so $k_s \leq T_i^s \leq n_s$.
The associated erasure probability $\epsilon_p$ is as in \eqref{eq:ep}.
Let $X_i$ denote the number of symbols that would be needed to receive the $i$-th transmitted packet correctly.
$X_i$ is distributed according to a negative binomial distribution with parameters $k_s$ and $1-\epsilon_s$.
The random variable $T_i^s$ is now given by $T_i^s = \min\{X_i,n_s\}$. As the sequence $(T_i^s)_{i=1}^{\infty}$ is i.i.d., one can use Wald's generalized equation\footnote{The general version of Wald's equation should be used because $T^p$ is a stopping time that cannot directly be defined as a function of $(T_i^s)_{i=1}^{\infty}$.} \cite[Theorem 5.5.2]{gallager2013discrete} to show that
\begin{align}
 \mathbb E\bigl[ T^{\text{FIR}} \bigr] = \mathbb E\bigl[ T^p \bigr] \mathbb E \bigl[ T_i^s \bigr].
\end{align}
$\mathbb E\bigl[ T^p \bigr]$ is the expected value of a negative binomial random variable with parameters $k_p$ and $1-\epsilon_p$.
The second expected value is given by
\begin{align}
 \mathbb E[T_i^s] &= \sum_{n=k_s}^{n_s} n {n-1 \choose k_s-1} (1-\epsilon_s)^{k_s} \epsilon_s^{n-k_s} \nonumber \\
 &+ n_s \sum_{n=n_s+1}^{\infty} {n-1 \choose k_s-1} (1-\epsilon_s)^{k_s} \epsilon_s^{n-k_s} 
\end{align}
and cannot be solved in closed form. Hence,
\begin{align}
 \mathbb E\bigl[ T^{\text{FIR}} \bigr] = \frac{k_p}{1-\epsilon_p} \mathbb E \bigl[ T_i^s \bigr].
\end{align}

As for FR, the performance of the FIR scheme depends on the choice of $n_s$. Choosing $n_s$ too small will result in a high packet erasure probability $\epsilon_p$. Letting $n_s \rightarrow \infty$ the scheme approaches the performance of the IIR scheme.
This is also shown in Fig.~\ref{fig:delay_somp_13}, which visualizes the dependency of the expected delay on the choice of $n_s$.
\begin{figure}[hbt]
\centering
\input{./figs/exp_T_over_nc_kc=100_eps_c=10perc.tex}
\caption{Normalized expected delay for FR and FIR as a function of $n_s$ for $\epsilon_s=0.1$ and $k_s = 100$.}
\label{fig:delay_somp_13}
\end{figure}
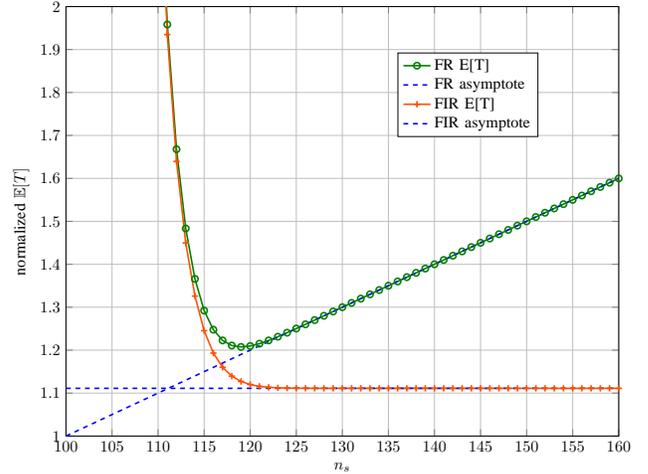

%% file: figs/delay_over_eps_c_kc=100.tex
%
%
%
\definecolor{mycolor1}{rgb}{0,0.498039215686275,0}%
\begin{tikzpicture}[scale=0.65]

\begin{axis}[%
width=4.82222222222222in,
height=2.2775in,
scale only axis,
xmin=0,
xmax=1,
xlabel={$\epsilon_s$},
xmajorgrids,
ymin=0,
ymax=25,
ylabel={optimal norm. $\mathbb E[T]$},
ymajorgrids,
axis x line*=bottom,
axis y line*=left,
legend style={at={(0.147470238095237,0.731301679060594)},anchor=south west,draw=black,fill=white,legend cell align=left}
]
\addplot [
color=mycolor1,
solid,
line width=1.0pt,
mark=o,
mark options={solid}
]
table[row sep=crcr]{
0 1\\
0.05 1.12090248132584\\
0.1 1.20746294329216\\
0.15 1.29796837618785\\
0.2 1.39542582699089\\
0.25 1.50385689267252\\
0.3 1.62605273839824\\
0.35 1.76550252495985\\
0.4 1.92723395875817\\
0.45 2.11694395734318\\
0.5 2.34367410352434\\
0.55 2.61990996410649\\
0.6 2.96431814106949\\
0.65 3.40615852573886\\
0.7 3.9944346381425\\
0.75 4.8169981180933\\
0.8 6.0497504461799\\
0.85 8.10302878993152\\
0.9 12.2076808002871\\
0.95 24.5183044296942\\
};
\addlegendentry{Fixed Redundancy (FR) };

\addplot [
color=red,
solid,
line width=1.0pt,
mark=+,
mark options={solid}
]
table[row sep=crcr]{
0 1\\
0.05 1.05263157894737\\
0.1 1.11111111111111\\
0.15 1.17647058823529\\
0.2 1.25\\
0.25 1.33333333333333\\
0.3 1.42857142857143\\
0.35 1.53846153846154\\
0.4 1.66666666666667\\
0.45 1.81818181818182\\
0.5 2\\
0.55 2.22222222222222\\
0.6 2.5\\
0.65 2.85714285714286\\
0.7 3.33333333333333\\
0.75 4\\
0.8 5\\
0.85 6.66666666666667\\
0.9 9.99999999999999\\
0.95 20\\
};
\addlegendentry{Infinite Incremental Redundancy (IIR) };

\end{axis}
\end{tikzpicture}%

%% file: figs/exp_T_over_nc_kc=100_eps_c=10perc.tex
%
%
%
\definecolor{mycolor1}{rgb}{0,0.498039215686275,0}%
\definecolor{mycolor2}{rgb}{1,0.315789473684211,0}%
\begin{tikzpicture}[scale=0.6]

\begin{axis}[%
width=4.82222222222222in,
height=3.749in,
scale only axis,
xmin=100,
xmax=160,
xlabel={$n_s$},
xmajorgrids,
ymin=1,
ymax=2,
ylabel={normalized $\mathbb E[T]$},
ymajorgrids,
legend style={at={(0.599925595238095,0.696665804002761)},anchor=south west,draw=black,fill=white,legend cell align=left}
]
\addplot [
color=mycolor1,
solid,
line width=1.0pt,
mark=o,
mark options={solid}
]
table[row sep=crcr]{
 110 2.4252991931408\\
111 1.95838095320779\\
112 1.66779842527401\\
113 1.48334940718986\\
114 1.36587384926394\\
115 1.29214979211458\\
116 1.24760816384735\\
117 1.22266057419746\\
118 1.2107813359361\\
119 1.20746294329216\\
120 1.20960571149978\\
121 1.21512140343061\\
122 1.22264723805384\\
123 1.23132738695839\\
124 1.24064618500609\\
125 1.25030568778411\\
126 1.26014066624759\\
127 1.27006302646396\\
128 1.28002752315391\\
129 1.29001172522077\\
130 1.30000487728876\\
131 1.31000198261534\\
132 1.32000078822483\\
133 1.33000030672164\\
134 1.34000011690558\\
135 1.35000004367373\\
136 1.36000001600229\\
137 1.37000000575426\\
138 1.38000000203187\\
139 1.39000000070493\\
140 1.40000000024042\\
141 1.41000000008065\\
142 1.42000000002662\\
143 1.43000000000865\\
144 1.44000000000277\\
145 1.45000000000087\\
146 1.46000000000027\\
147 1.47000000000008\\
148 1.48000000000003\\
149 1.49000000000001\\
150 1.5\\
151 1.51\\
152 1.52\\
153 1.53\\
154 1.54\\
155 1.55\\
156 1.56\\
157 1.57\\
158 1.58\\
159 1.59\\
160 1.6\\
};
\addlegendentry{FR E[T]};

\addplot [
color=blue,
dashed,
line width=1.0pt
]
table[row sep=crcr]{
100 1\\
101 1.01\\
102 1.02\\
103 1.03\\
104 1.04\\
105 1.05\\
106 1.06\\
107 1.07\\
108 1.08\\
109 1.09\\
110 1.1\\
111 1.11\\
112 1.12\\
113 1.13\\
114 1.14\\
115 1.15\\
116 1.16\\
117 1.17\\
118 1.18\\
119 1.19\\
120 1.2\\
121 1.21\\
122 1.22\\
123 1.23\\
124 1.24\\
125 1.25\\
126 1.26\\
127 1.27\\
128 1.28\\
129 1.29\\
130 1.3\\
131 1.31\\
132 1.32\\
133 1.33\\
134 1.34\\
135 1.35\\
136 1.36\\
137 1.37\\
138 1.38\\
139 1.39\\
140 1.4\\
141 1.41\\
142 1.42\\
143 1.43\\
144 1.44\\
145 1.45\\
146 1.46\\
147 1.47\\
148 1.48\\
149 1.49\\
150 1.5\\
151 1.51\\
152 1.52\\
153 1.53\\
154 1.54\\
155 1.55\\
156 1.56\\
157 1.57\\
158 1.58\\
159 1.59\\
160 1.6\\
};
\addlegendentry{FR asymptote};

\addplot [
color=mycolor2,
solid,
line width=1.0pt,
mark=+,
mark options={solid}
]
table[row sep=crcr]{
110 2.40589399348813\\
111 1.93485074916447\\
112 1.63949835361543\\
113 1.44958655957232\\
114 1.32593033301072\\
115 1.24531286454819\\
116 1.19320354271648\\
117 1.16008326891556\\
118 1.13951826995805\\
119 1.12710360661118\\
120 1.11984035354465\\
121 1.11572939129292\\
122 1.11347991753101\\
123 1.11228975972663\\
124 1.11168049081298\\
125 1.11137840685848\\
126 1.11123317445929\\
127 1.11116538786495\\
128 1.11113463451003\\
129 1.11112105715261\\
130 1.1111152174263\\
131 1.11111276792868\\
132 1.11111176494448\\
133 1.11111136366894\\
134 1.11111120666978\\
135 1.11111114655103\\
136 1.11111112400285\\
137 1.11111111571362\\
138 1.11111111272472\\
139 1.11111111166697\\
140 1.11111111129936\\
141 1.11111111117382\\
142 1.11111111113167\\
143 1.11111111111774\\
144 1.11111111111322\\
145 1.11111111111177\\
146 1.11111111111132\\
147 1.11111111111117\\
148 1.11111111111113\\
149 1.11111111111112\\
150 1.11111111111111\\
151 1.11111111111111\\
152 1.11111111111111\\
153 1.11111111111111\\
154 1.11111111111111\\
155 1.11111111111111\\
156 1.11111111111111\\
157 1.11111111111111\\
158 1.11111111111111\\
159 1.11111111111111\\
160 1.11111111111111\\
};
\addlegendentry{FIR E[T]};

\addplot [
color=blue,
dashed,
line width=1.0pt
]
table[row sep=crcr]{
100 1.11111111111111\\
101 1.11111111111111\\
102 1.11111111111111\\
103 1.11111111111111\\
104 1.11111111111111\\
105 1.11111111111111\\
106 1.11111111111111\\
107 1.11111111111111\\
108 1.11111111111111\\
109 1.11111111111111\\
110 1.11111111111111\\
111 1.11111111111111\\
112 1.11111111111111\\
113 1.11111111111111\\
114 1.11111111111111\\
115 1.11111111111111\\
116 1.11111111111111\\
117 1.11111111111111\\
118 1.11111111111111\\
119 1.11111111111111\\
120 1.11111111111111\\
121 1.11111111111111\\
122 1.11111111111111\\
123 1.11111111111111\\
124 1.11111111111111\\
125 1.11111111111111\\
126 1.11111111111111\\
127 1.11111111111111\\
128 1.11111111111111\\
129 1.11111111111111\\
130 1.11111111111111\\
131 1.11111111111111\\
132 1.11111111111111\\
133 1.11111111111111\\
134 1.11111111111111\\
135 1.11111111111111\\
136 1.11111111111111\\
137 1.11111111111111\\
138 1.11111111111111\\
139 1.11111111111111\\
140 1.11111111111111\\
141 1.11111111111111\\
142 1.11111111111111\\
143 1.11111111111111\\
144 1.11111111111111\\
145 1.11111111111111\\
146 1.11111111111111\\
147 1.11111111111111\\
148 1.11111111111111\\
149 1.11111111111111\\
150 1.11111111111111\\
151 1.11111111111111\\
152 1.11111111111111\\
153 1.11111111111111\\
154 1.11111111111111\\
155 1.11111111111111\\
156 1.11111111111111\\
157 1.11111111111111\\
158 1.11111111111111\\
159 1.11111111111111\\
160 1.11111111111111\\
};
\addlegendentry{FIR asymptote};

\end{axis}
\end{tikzpicture}%

%% file: RAms.tex
\section{Multicast Scenarios}
\label{sec:ms}
We are now concerned with multicast  to $u$ users, each independently experiencing a channel with erasure probability $\epsilon_s$.  As in the point-to-point case, we distinguish between transmission
scenarios depending on the type of coding used at the (inner) physical layer. Again, since the objective is
ultimate decoding of the file (this time by all users), rateless coding is necessary at either physical or 
packet level, but now coded symbols or packets have to be transmitted until all users are able to decode. 
We first analyze and then compare the IIR and FR schemes.

\subsection{IIR with $u$ Users}
Let $T_i^s(j)$  be the number of coded symbol-transmissions until user $j$, $j\in\{1,\dots,u\}$ is able to decode packet $i$, $i\in\{1,\dots,k_p\}$. Note that 
$T_i^s(j)$ are i.i.d.\ $NB(k_s,1-\epsilon_s)$.
Then the  number of coded symbol-transmissions until packet $i$ is decoded by all users,
is given by 
\[
T_i^s = \max_{j\in\{1,\dots,u\} } T_i^s(j),
\]
and is  therefore distributed as the maximum order statistic of $u$ of 
i.i.d.\ $NB(k_s,1-\epsilon_s)$.

In the IIR scheme, multicast of packet $i$  starts only when all users decode packet $i-1$, and
a chunk is decoded when all packets are decoded by all users. Therefore, the
number of symbol-transmissions until a chunk is decoded is given by 
\[
T^{\text{IIR}} = \sum_{i=1}^{k_p} T_i^s = 
\sum_{i=1}^{k_p} \max_{j\in\{1,\dots,u\} } T_i^s(j)
\]
and its average by
\[
\mathbb{E}\bigl[T^{\text{IIR}}\bigr]= k_p\cdot \mathbb{E}\bigl[T_1^s\bigr].
\]
\subsection{FR with $u$ Users}
In the FR scheme, each packet takes exactly $n_s$ symbol transmissions.
At that point some users will not be able to decode the packet, and that will happen with probability $\epsilon_p$ given by (\ref{eq:ep}).
Let $T^p(j)$  be the number of coded packet-transmissions until user $j$, $j\in\{1,\dots,u\}$ is able to decode the chunk. Note that  $T^p(j)$ are i.i.d.\
$NB(k_p,1-\epsilon_p)$. Therefore, the number of packet-transmissions until a chunk is decoded is given by
\[
T^p = \max_{j\in\{1,\dots,u\} } T^p(j)
\]
and the average number of symbol-transmissions until a chunk is decoded is given by 
\[
\mathbb{E}\bigl[T^{\text{FR}}\bigr]= n_s\cdot \mathbb{E}\bigl[T_p\bigr]
\]
\subsection{IIR vs.\ FR with $u$ Users}
We have seen that in the point-to-point case, coding ratelessly at the symbol level results in fewer transmissions on average. As we will see, that is not always the case for the multiple users. Earlier in this section, we have seen the following:
\begin{itemize}
\item The IIR scheme chunk download time (measured by the number of symbol-transmissions) is equal to the
``maximum order statistic of $u$ of $NB(k_s,1-\epsilon_s)$''  $\times$ $k_p$.
\item The FR scheme chunk download time (measured by the number of symbol-transmissions) is equal to the ``maximum order statistic of $u$ of $NB(k_p,1-\epsilon_p)$'' $\times$ $n_s$.
\end{itemize}

The expected value of the maximum order statistic of negative binomial random variables with parameters $(k,1-\epsilon)$ has been investigated in \cite{grabner1997maximum}. One result of this work is the following approximation:
\[
\log_{1/\epsilon}u+(k-1)
   \log_{1/\epsilon}\bigl[\log_{1/\epsilon}u\bigr]
\]
We found in our numerical simulations that this approximation is not precise enough in many cases. 
Instead, we found that the expectation  of the maximum order statistic of $u$ of $NB(k,1-\epsilon)$ behaves as follows (see Appendix for an argument):
  \begin{itemize}
  \item A good numerical approximation when $k$ is small is
 \[
 \log_{1/\epsilon}u+(k-1)
   \bigl[\log_{1/\epsilon}\log_{1/\epsilon}u+\log_{1/\epsilon}(1-\epsilon)/\epsilon\bigr]
 \]
 which also gives good qualitative description for larger $k$ that can be used to get an insight into
 the behaviour (if not the precise value) of the chunk download time for each of the schemes.
 \item  $k/(1-\epsilon)$ is a good approximation  when $k$ is very large or when $k$ is large compared to $u$, and also  when $\epsilon$ is very small (but for different probabilistic  reasons than when $k$ is large, see Appendix).
\end{itemize}

Our approximation indicates that the expected value of $u$ of $NB(k,1-\epsilon)$ will
\begin{enumerate}
\item increase with $u$,
\item have a lower rate of increase with $u$ as $k$ gets larger,
\item have a highr rate of increase with $u$ as $\epsilon$ gets larger.
\end{enumerate}
Figures~\ref{fig:MultiUser_EqualK} and \ref{fig:MultiUser_LargeKS} obtained by simulation clearly illustrate this behaviour. 
\begin{figure}[hbt]
\centering
\input{./figs/MultiUser_ET_kp=100_epsc=01-05_kc=100.tex}
\caption{Normalized expected delay for FR and IIR as a function of the number of users for $k_p = 100$ and $k_s = 100$, $\epsilon_s\in\{0.1, 0.2, 0.3.0.4,0.5\}$.}
\label{fig:MultiUser_EqualK}
\end{figure}
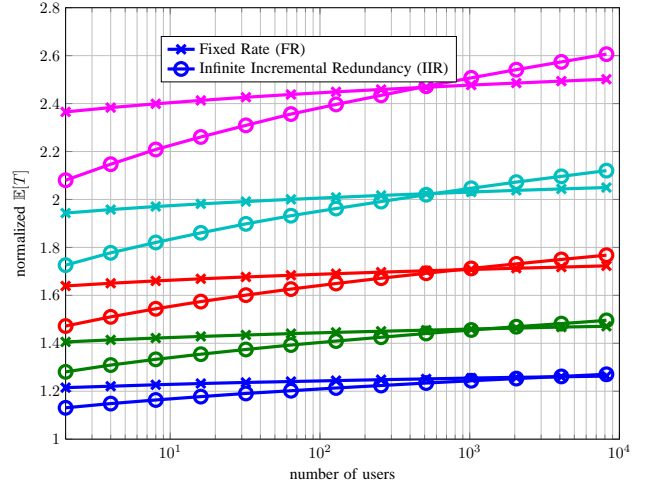
\begin{figure}[hbt]
\centering
\input{./figs/MultiUser_ET_kp=100_epsc=01-05_kc=1000.tex}
\caption{Normalized expected delay for FR and IIR as a function of the number of users for $k_p = 100$ and $k_s = 1000$, $\epsilon_s\in\{0.1, 0.2, 0.3.0.4,0.5\}$.}
\label{fig:MultiUser_LargeKS}
\end{figure}
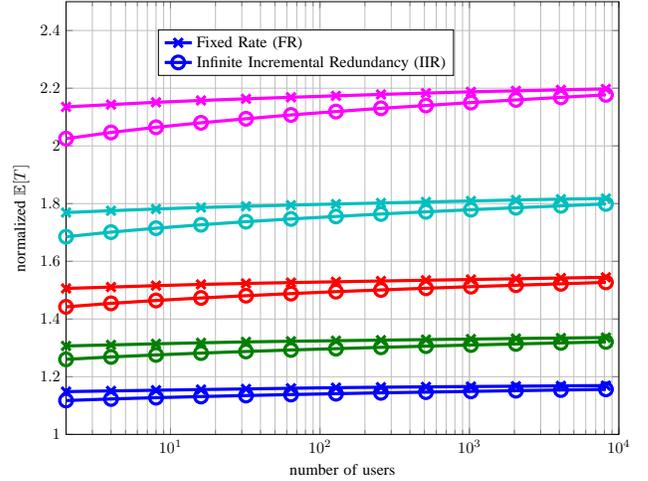

As can be seen from the figures, for a single user or a small number of users, the IIR scheme has a lower expected download time. This encourages the use of Hybrid-ARQ only. However, for a larger group of users, it is better to admit a certain packet erasure probability and rely on packet level coding for reliable transmission.

This result might be counter-intuitive, but can be explained in the following way:
Both the expected download time of IIR and FR involve the expected value of the maximum of negative binomial random variables. Only the parameters are different.
It is thus not immediately clear why the performance curves show the observed behavior.
One possible explanation is as follows:
For a large number of users, the time needed to satisfy \emph{all} $u$ users in the IIR scheme depends on the value of $\epsilon_s$. Similarly, in the FR scheme, the number of packets needed to satisfy all $u$ users depends on $\epsilon_p$. The value of $\epsilon_p$ can be influenced through the optimal choice of $n_s$ whereas $\epsilon_s$ cannot be changed.
We have thus the possibility to adapt the coin flipping bias associated with the packet transmission - and this leads to a smaller expected download time for FR as the number of users large number of users gets large.

%% file: figs/MultiUser_ET_kp=100_epsc=01-05_kc=100.tex
%
%
%
\definecolor{mycolor1}{rgb}{0,0.498039215686275,0}%
\definecolor{mycolor2}{rgb}{0,0.749019607843137,0.749019607843137}%
\definecolor{mycolor3}{rgb}{1,0,1}%
\begin{tikzpicture}[scale=0.6]

\begin{axis}[%
width=4.83083333333333in,
height=3.76666666666667in,
scale only axis,
xmode=log,
xmin=2,
xmax=10000,
xminorticks=true,
xlabel={number of users},
xmajorgrids,
xminorgrids,
ymin=1,
ymax=2.8,
ylabel={normalized $\mathbb E[T]$},
ymajorgrids,
legend style={at={(0.17217789104278,0.827334504897574)},anchor=south west,draw=black,fill=white,legend cell align=left}
]
\addplot [
color=blue,
solid,
line width=2.0pt,
mark size=4.0pt,
mark=x,
mark options={solid}
]
table[row sep=crcr]{
2 1.2150228\\
4 1.2208788\\
8 1.2268536\\
16 1.2321612\\
32 1.23662242\\
64 1.24043997\\
128 1.24433012\\
256 1.24815493\\
512 1.25131496\\
1024 1.25495788\\
2048 1.25795908\\
4096 1.26038444\\
8192 1.2629684\\
};
\addlegendentry{Fixed Rate (FR)};

\addplot [
color=blue,
solid,
line width=2.0pt,
mark size=4.0pt,
mark=o,
mark options={solid}
]
table[row sep=crcr]{
2 1.130825\\
4 1.148417\\
8 1.163465\\
16 1.177507\\
32 1.190668\\
64 1.202003\\
128 1.213524\\
256 1.223751\\
512 1.234089\\
1024 1.243656\\
2048 1.252723\\
4096 1.261705\\
8192 1.270722\\
};
\addlegendentry{Infinite Incremental Redundancy (IIR)};

\addplot [
color=mycolor1,
solid,
line width=2.0pt,
mark size=4.0pt,
mark=x,
mark options={solid},
forget plot
]
table[row sep=crcr]{
2 1.40566384\\
4 1.41409842\\
8 1.42179882\\
16 1.42856138\\
32 1.43450641\\
64 1.44042503\\
128 1.445612\\
256 1.4502446\\
512 1.4548422\\
1024 1.45928796\\
2048 1.46334312\\
4096 1.4676942\\
8192 1.47083304\\
};
\addplot [
color=mycolor1,
solid,
line width=2.0pt,
mark size=4.0pt,
mark=o,
mark options={solid},
forget plot
]
table[row sep=crcr]{
2 1.280841\\
4 1.309173\\
8 1.332988\\
16 1.354436\\
32 1.374167\\
64 1.392675\\
128 1.409259\\
256 1.425208\\
512 1.440951\\
1024 1.455513\\
2048 1.469162\\
4096 1.482035\\
8192 1.495514\\
};
\addplot [
color=red,
solid,
line width=2.0pt,
mark size=4.0pt,
mark=x,
mark options={solid},
forget plot
]
table[row sep=crcr]{
2 1.6391952\\
4 1.650376\\
8 1.66021268\\
16 1.66902421\\
32 1.67669028\\
64 1.68379074\\
128 1.69041168\\
256 1.69682348\\
512 1.70238993\\
1024 1.70792378\\
2048 1.71306528\\
4096 1.717818\\
8192 1.723095\\
};
\addplot [
color=red,
solid,
line width=2.0pt,
mark size=4.0pt,
mark=o,
mark options={solid},
forget plot
]
table[row sep=crcr]{
2 1.471957\\
4 1.510273\\
8 1.543979\\
16 1.573998\\
32 1.60089\\
64 1.62619\\
128 1.649329\\
256 1.671501\\
512 1.692196\\
1024 1.712176\\
2048 1.731271\\
4096 1.749814\\
8192 1.767501\\
};
\addplot [
color=mycolor2,
solid,
line width=2.0pt,
mark size=4.0pt,
mark=x,
mark options={solid},
forget plot
]
table[row sep=crcr]{
2 1.9433169\\
4 1.95783966\\
8 1.9705489\\
16 1.98153905\\
32 1.99131252\\
64 2.00039232\\
128 2.00874624\\
256 2.01692527\\
512 2.02394275\\
1024 2.03112374\\
2048 2.03791762\\
4096 2.0439452\\
8192 2.04978848\\
};
\addplot [
color=mycolor2,
solid,
line width=2.0pt,
mark size=4.0pt,
mark=o,
mark options={solid},
forget plot
]
table[row sep=crcr]{
2 1.725976\\
4 1.777404\\
8 1.820608\\
16 1.860774\\
32 1.898174\\
64 1.932173\\
128 1.962075\\
256 1.991157\\
512 2.019742\\
1024 2.046681\\
2048 2.072419\\
4096 2.096419\\
8192 2.120311\\
};
\addplot [
color=mycolor3,
solid,
line width=2.0pt,
mark size=4.0pt,
mark=x,
mark options={solid},
forget plot
]
table[row sep=crcr]{
2 2.36532444\\
4 2.3832738\\
8 2.39904588\\
16 2.41338174\\
32 2.42638592\\
64 2.43803045\\
128 2.44892786\\
256 2.4589188\\
512 2.46815244\\
1024 2.4774282\\
2048 2.4856091\\
4096 2.4937294\\
8192 2.50160847\\
};
\addplot [
color=mycolor3,
solid,
line width=2.0pt,
mark size=4.0pt,
mark=o,
mark options={solid},
forget plot
]
table[row sep=crcr]{
2 2.080202\\
4 2.147045\\
8 2.208485\\
16 2.260704\\
32 2.309155\\
64 2.356264\\
128 2.396406\\
256 2.434324\\
512 2.472179\\
1024 2.507337\\
2048 2.541509\\
4096 2.573887\\
8192 2.606159\\
};
\end{axis}
\end{tikzpicture}%

%% file: figs/MultiUser_ET_kp=100_epsc=01-05_kc=1000.tex
%
%
%
\definecolor{mycolor1}{rgb}{0,0.498039215686275,0}%
\definecolor{mycolor2}{rgb}{0,0.749019607843137,0.749019607843137}%
\definecolor{mycolor3}{rgb}{1,0,1}%
\begin{tikzpicture}[scale=0.6]

\begin{axis}[%
width=4.82222222222222in,
height=3.77777777777778in,
scale only axis,
xmode=log,
xmin=2,
xmax=10000,
xminorticks=true,
xlabel={number of users},
xmajorgrids,
xminorgrids,
ymin=1,
ymax=2.5,
ylabel={normalized $\mathbb E[T]$},
ymajorgrids,
legend style={at={(0.166517857142856,0.828869047619054)},anchor=south west,draw=black,fill=white,legend cell align=left}
]
\addplot [
color=blue,
solid,
line width=2.0pt,
mark size=4.0pt,
mark=x,
mark options={solid}
]
table[row sep=crcr]{
2 1.148256657\\
4 1.150641725\\
8 1.153039884\\
16 1.15540416\\
32 1.15750823\\
64 1.159874912\\
128 1.161941832\\
256 1.16352975\\
512 1.164738816\\
1024 1.165903488\\
2048 1.16701135\\
4096 1.16820627\\
8192 1.169187648\\
};
\addlegendentry{Fixed Rate (FR)};

\addplot [
color=blue,
solid,
line width=2.0pt,
mark size=4.0pt,
mark=o,
mark options={solid}
]
table[row sep=crcr]{
2 1.1173629\\
4 1.1227975\\
8 1.1272951\\
16 1.1311616\\
32 1.1347112\\
64 1.1380739\\
128 1.1410679\\
256 1.1440375\\
512 1.1464936\\
1024 1.1491021\\
2048 1.15151\\
4096 1.1537949\\
8192 1.1561174\\
};
\addlegendentry{Infinite Incremental Redundancy (IIR)};

\addplot [
color=mycolor1,
solid,
line width=2.0pt,
mark size=4.0pt,
mark=x,
mark options={solid},
forget plot
]
table[row sep=crcr]{
2 1.3065884\\
4 1.310235559\\
8 1.313930115\\
16 1.317490848\\
32 1.32060519\\
64 1.322819974\\
128 1.324776068\\
256 1.326674316\\
512 1.328679456\\
1024 1.330453908\\
2048 1.332402624\\
4096 1.333953072\\
8192 1.335715687\\
};
\addplot [
color=mycolor1,
solid,
line width=2.0pt,
mark size=4.0pt,
mark=o,
mark options={solid},
forget plot
]
table[row sep=crcr]{
2 1.2600818\\
4 1.2685263\\
8 1.2754714\\
16 1.2817796\\
32 1.2873776\\
64 1.2925966\\
128 1.2971834\\
256 1.3018002\\
512 1.3058774\\
1024 1.3097799\\
2048 1.3134763\\
4096 1.3171309\\
8192 1.3206193\\
};
\addplot [
color=red,
solid,
line width=2.0pt,
mark size=4.0pt,
mark=x,
mark options={solid},
forget plot
]
table[row sep=crcr]{
2 1.505912496\\
4 1.510699462\\
8 1.515574784\\
16 1.51987472\\
32 1.523336928\\
64 1.526232334\\
128 1.52899425\\
256 1.531728576\\
512 1.53438291\\
1024 1.536921244\\
2048 1.5393496\\
4096 1.542023793\\
8192 1.54439495\\
};
\addplot [
color=red,
solid,
line width=2.0pt,
mark size=4.0pt,
mark=o,
mark options={solid},
forget plot
]
table[row sep=crcr]{
2 1.4421454\\
4 1.4541979\\
8 1.4641489\\
16 1.4729371\\
32 1.4807302\\
64 1.4879819\\
128 1.494792\\
256 1.5007421\\
512 1.5065792\\
1024 1.5120127\\
2048 1.5170316\\
4096 1.5219698\\
8192 1.5269592\\
};
\addplot [
color=mycolor2,
solid,
line width=2.0pt,
mark size=4.0pt,
mark=x,
mark options={solid},
forget plot
]
table[row sep=crcr]{
2 1.76895927\\
4 1.775376804\\
8 1.78123792\\
16 1.786432788\\
32 1.790725766\\
64 1.794770362\\
128 1.798548612\\
256 1.802081232\\
512 1.805725623\\
1024 1.80916533\\
2048 1.81259696\\
4096 1.815345504\\
8192 1.817876487\\
};
\addplot [
color=mycolor2,
solid,
line width=2.0pt,
mark size=4.0pt,
mark=o,
mark options={solid},
forget plot
]
table[row sep=crcr]{
2 1.6851649\\
4 1.701141\\
8 1.71465\\
16 1.7265489\\
32 1.7372617\\
64 1.7467939\\
128 1.7553745\\
256 1.76395\\
512 1.771498\\
1024 1.7788726\\
2048 1.7856108\\
4096 1.7921001\\
8192 1.7987846\\
};
\addplot [
color=mycolor3,
solid,
line width=2.0pt,
mark size=4.0pt,
mark=x,
mark options={solid},
forget plot
]
table[row sep=crcr]{
2 2.135244306\\
4 2.14326244\\
8 2.15110852\\
16 2.157496208\\
32 2.163184764\\
64 2.16849924\\
128 2.17355853\\
256 2.178450048\\
512 2.183079\\
1024 2.187524632\\
2048 2.190851047\\
4096 2.194278562\\
8192 2.197447056\\
};
\addplot [
color=mycolor3,
solid,
line width=2.0pt,
mark size=4.0pt,
mark=o,
mark options={solid},
forget plot
]
table[row sep=crcr]{
2 2.0251501\\
4 2.0465435\\
8 2.0645762\\
16 2.0800666\\
32 2.0940213\\
64 2.1072872\\
128 2.1190731\\
256 2.1302422\\
512 2.1400713\\
1024 2.1501773\\
2048 2.1593686\\
4096 2.1682663\\
8192 2.1769976\\
};
\end{axis}
\end{tikzpicture}%

%% file: MaxNegBinom.tex
We here study the expected value of the maximum order statistic of $u$ of $NB(k,1-\epsilon)$.
Consider a game with $u$ players, each independently tossing a biased coin with the probability of head equal to $(1-\epsilon)$ until he sees $k$ heads. 
We separately analyze three cases:
\subsubsection{Small $k$, Large $u$}
The probability that a players did not see any heads after $t$ tosses is equal to $\epsilon^t$.
Therefore, the expected number of such players is $u\cdot\epsilon^t$, which is smaller than
$1$ when $t>t_1$, where
\[
t_1=\log_{1/\epsilon}u.
\]
The probability that a players saw only one head after $t$ tosses is equal to 
$t\cdot\epsilon^{(t-1)}(1-\epsilon)$.
Therefore, the expected number of such players is $u\cdot t\cdot\epsilon^{(t-1)}(1-\epsilon)$, which becomes equal to $t_1\cdot(1-\epsilon)/\epsilon$ when $t=t_1$. 
By the same argument as before, these players need another $\Delta t = \log_{1/\epsilon}  t_1\cdot(1-\epsilon)/\epsilon$ tosses on average until each of these users has seen at least one more head on average. 
Therefore, after $t_1 + 
\log_{1/\epsilon}t_1\cdot(1-\epsilon)$ tosses the average, the number of users who have seen fever than $2$ heads is smaller than 1.  

Let $t_k$ denote 
the moment at which the expected number of players who have seen $k$ heads is smaller than $1$. 
By the above reasoning, we have that 
\[
t_2 = t_1 + \Delta t = \log_{1/\epsilon}u+\log_{1/\epsilon}\log_{1/\epsilon}u+\log_{1/\epsilon}(1-\epsilon)/\epsilon
\]
and for small $k$, we have
\[
t_k\approx  \log_{1/\epsilon}u+(k-1)
   \bigl[\log_{1/\epsilon}\log_{1/\epsilon}u+\log_{1/\epsilon}(1-\epsilon)/\epsilon\bigr].
 \]
 
 \subsubsection{Large $k$}
 For large $k$, the number of tosses that the players have to make is also large, and by the law of large numbers, each player will take approximately $k/(1-\epsilon)$ to see $k$ heads. 
  \subsubsection{Small $\epsilon$}
  When $\epsilon$ is small, then $1-\epsilon$ is large, and thus almost all users see a head each time they toss their coins. 
  Note that, as predicted by 2) and 3) above, the dependence of the download time on the number of users 
  is small when $k$ is large or $\epsilon$ is small.
  
  Figure~\ref{fig:OrderStatApprox} illustrates the tightness of our bound and the bound of \cite{grabner1997maximum}.
\begin{figure}[hbt]
\centering
\input{./figs/new_order_stat_plots.tex}
\caption{Order statistics approximations and simulation results for $\epsilon=0.1$.
Both approximations are shown for different values of $k$. Dashed lines show the approximation that holds if $k \gg u$. For $k=1$, both approximations match.}
\label{fig:OrderStatApprox}
\end{figure}
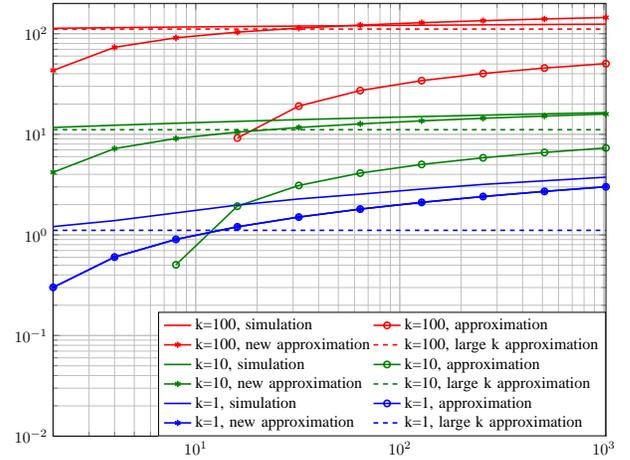

%% file: figs/new_order_stat_plots.tex
%
%
\begin{tikzpicture}[scale=0.6]

\begin{axis}[%
width=4.82222222222222in,
height=3.77777777777778in,
scale only axis,
xmode=log,
xmin=2,
xmax=1024,
xminorticks=true,
xmajorgrids,
xminorgrids,
ymode=log,
ymin=0.01,
ymax=200,
yminorticks=true,
ymajorgrids,
yminorgrids,
legend columns=2,
legend style={at={(1,0.0)},anchor=south east,draw=black,fill=white,legend cell align=left}
]

\addplot [
color=red,
solid,
line width=1.0pt
]
table[row sep=crcr]{
2 113.08219\\
4 114.81536\\
8 116.37314\\
16 117.75774\\
32 119.03178\\
64 120.24182\\
128 121.35954\\
256 122.40433\\
512 123.40968\\
1024 124.37281\\
};
\addlegendentry{k=100, simulation};

\addplot [
color=red,
solid,
line width=1.0pt,
mark=o,
mark options={solid}
]
table[row sep=crcr]{
2 -51.3166025421142\\
4 -21.2136029757161\\
8 -3.47953833353964\\
16 9.19042658634604\\
32 19.0855478698076\\
64 27.2255212241864\\
128 34.1542833932811\\
256 40.1965161397361\\
512 45.5616458576908\\
1024 50.3926674188616\\
};
\addlegendentry{k=100, approximation};

\addplot [
color=red,
solid,
line width=1.0pt,
mark=asterisk,
mark options={solid}
]
table[row sep=crcr]{
2 43.153405892379\\
4 73.2564054587771\\
8 90.9904701009535\\
16 103.660435020839\\
32 113.555556304301\\
64 121.69552965868\\
128 128.624291827774\\
256 134.666524574229\\
512 140.031654292184\\
1024 144.862675853355\\
};
\addlegendentry{k=100, new approximation};

\addplot [
color=red,
dashed,
line width=1.0pt
]
table[row sep=crcr]{
2 111.111111111111\\
4 111.111111111111\\
8 111.111111111111\\
16 111.111111111111\\
32 111.111111111111\\
64 111.111111111111\\
128 111.111111111111\\
256 111.111111111111\\
512 111.111111111111\\
1024 111.111111111111\\
};
\addlegendentry{k=100, large k approximation};

\addplot [
color=green!50!black,
solid,
line width=1.0pt
]
table[row sep=crcr]{
2 11.69105\\
4 12.30296\\
8 12.89023\\
16 13.45461\\
32 13.99116\\
64 14.51656\\
128 15.01451\\
256 15.49561\\
512 15.97423\\
1024 16.43166\\
};
\addlegendentry{k=10, simulation};

\addplot [
color=green!50!black,
solid,
line width=1.0pt,
mark=o,
mark options={solid}
]
table[row sep=crcr]{
2 -4.39148205322494\\
4 -1.38118209658513\\
8 0.504669230579982\\
16 1.93014785571866\\
32 3.10336796845515\\
64 4.11702917854776\\
128 5.02058028088726\\
256 5.84353779935042\\
512 6.60494049704083\\
1024 7.31778790775089\\
};
\addlegendentry{k=10, approximation};

\addplot [
color=green!50!black,
solid,
line width=1.0pt,
mark=asterisk,
mark options={solid}
]
table[row sep=crcr]{
2 4.19670053172898\\
4 7.20700048836879\\
8 9.0928518155339\\
16 10.5183304406726\\
32 11.6915505534091\\
64 12.7052117635017\\
128 13.6087628658412\\
256 14.4317203843043\\
512 15.1931230819948\\
1024 15.9059704927048\\
};
\addlegendentry{k=10, new approximation};

\addplot [
color=green!50!black,
dashed,
line width=1.0pt
]
table[row sep=crcr]{
2 11.1111111111111\\
4 11.1111111111111\\
8 11.1111111111111\\
16 11.1111111111111\\
32 11.1111111111111\\
64 11.1111111111111\\
128 11.1111111111111\\
256 11.1111111111111\\
512 11.1111111111111\\
1024 11.1111111111111\\
};
\addlegendentry{k=10, large k approximation};


\addplot [
color=blue,
solid,
line width=1.0pt
]
table[row sep=crcr]{
2 1.21001\\
4 1.38733\\
8 1.65548\\
16 1.97994\\
32 2.27831\\
64 2.54061\\
128 2.85593\\
256 3.17827\\
512 3.45087\\
1024 3.74811\\
};
\addlegendentry{k=1, simulation};

\addplot [
color=blue,
solid,
line width=1.0pt,
mark=o,
mark options={solid}
]
table[row sep=crcr]{
2 0.301\\
4 0.602\\
8 0.903\\
16 1.204\\
32 1.505\\
64 1.806\\
128 2.107\\
256 2.408\\
512 2.709\\
1024 3.01\\
};
\addlegendentry{k=1, approximation};

\addplot [
color=blue,
solid,
line width=1.0pt,
mark=asterisk,
mark options={solid}
]
table[row sep=crcr]{
2 0.301029995663981\\
4 0.602059991327962\\
8 0.903089986991943\\
16 1.20411998265592\\
32 1.50514997831991\\
64 1.80617997398389\\
128 2.10720996964787\\
256 2.40823996531185\\
512 2.70926996097583\\
1024 3.01029995663981\\
};
\addlegendentry{k=1, new approximation};

\addplot [
color=blue,
dashed,
line width=1.0pt
]
table[row sep=crcr]{
2 1.11111111111111\\
4 1.11111111111111\\
8 1.11111111111111\\
16 1.11111111111111\\
32 1.11111111111111\\
64 1.11111111111111\\
128 1.11111111111111\\
256 1.11111111111111\\
512 1.11111111111111\\
1024 1.11111111111111\\
};
\addlegendentry{k=1, large k approximation};

\end{axis}
\end{tikzpicture}%